\documentclass[a4paper]{article}

\usepackage{amsmath,amssymb,amsthm,graphicx,fixmath,tikz,latexsym,color,xspace}
\usepackage[margin=1in]{geometry}

\usepackage[ruled,noend,linesnumbered]{algorithm2e}
\usepackage{xspace}
\usepackage{wrapfig}
\usepackage{enumerate}
\usepackage{url}

\definecolor{Apricot}{rgb}{0.98, 0.81, 0.69}
\definecolor{Orange}{rgb}{1.0, 0.5, 0.0}

\newtheorem{theorem}{Theorem}[section]
\newtheorem{lemma}[theorem]{Lemma}

\newtheorem{observation}[theorem]{Observation}

\newcommand*\samethanks[1][\value{footnote}]{\footnotemark[#1]}

\begin{document}

\title{Hanabi is NP-complete, Even for Cheaters who Look at Their Cards\thanks{Y.~D.~was supported by the IMPACT Tough Robotics Challenge Project of Japan Science and Technology Agency.} $^{,}$\thanks{M.~K.~was supported in part by the ELC project (MEXT KAKENHI No.~12H00855 and 15H02665).} $^{,}$\thanks{V.~M~was supported by the ERC Starting Grant PARAMTIGHT (No.~280152).}}

\author{Jean-Francois Baffier\thanks{National Institute of Informatics (NII), Tokyo, Japan. \texttt{ \{jf\_baffier,chiumk,andre,marcel\}@nii.ac.jp } } $^{,}$\thanks{JST, ERATO, Kawarabayashi Large Graph Project.} 
\hspace{.3in} Man-Kwun Chiu\samethanks[4] $^{,}$\samethanks[5]
\hspace{.3in} Yago Diez\samethanks[3] $^{,}$\thanks{Tohoku University, Sendai, Japan. \texttt{ \{yago,mati\}@dais.is.tohoku.ac.jp }} 
\hspace{.3in} Matias Korman\samethanks[1] $^{,}$\samethanks[6] \\
Valia Mitsou\samethanks[2] $^{,}$\thanks{SZTAKI, Hungarian Academy of Sciences.  \texttt{ vmitsou@sztaki.hu}} 
\hspace{.3in} Andr\'e van Renssen\samethanks[4] $^{,}$\samethanks[5]
\hspace{.3in} Marcel Roeloffzen\samethanks[4] $^{,}$\samethanks[5]
\hspace{.3in} Yushi Uno\thanks{Graduate School of Science, Osaka Prefecture University. \texttt{ uno@mi.s.osakafu-u.ac.jp  }}
}
\date{}

\maketitle

\begin{abstract}
In this paper we study a cooperative card game called Hanabi from the viewpoint of algorithmic combinatorial game theory. In Hanabi, each card has one among $c$ colors and a number between $1$ and $n$. The aim is to make, for each color, a pile of cards of that color with all increasing numbers from $1$ to $n$.
At each time during the game, each player holds $h$ cards in hand. Cards are drawn sequentially from a deck and the players should decide whether to play, discard or store them for future use. One of the features of the game is that the players can see their partners' cards but not their own and information must be shared through hints.

We introduce a single-player, perfect-information model and show that the game is intractable even for this simplified version where we forego both the hidden information and the multiplayer aspect of the game, even when the player can only hold two cards in her hand.
On the positive side, we show that the decision version of the problem---to decide whether or not numbers from $1$ through $n$ can be played for every color---can be solved in (almost) linear time for some restricted cases.
\end{abstract}


\section{Introduction}

When studying mathematical puzzles or games,
mathematicians and computer scientists are usually interested in finding winning strategies,
often by trying to design computer programs to play as close to optimally as possible. 
However, their effort is encumbered by the combinatorial explosion of the available choices in subsequent rounds.
The field of computational complexity provides tools to help decide whether solving a given puzzle
or finding the winner of a given game can be done efficiently, or to give strong evidence that such tasks might be practically infeasible.
Many games and puzzles have been studied from this perspective. Extensive lists can be found in \cite{bob+erik09} as well as in \cite{npc}. Some concrete examples include classic games, like Hex \cite{hex}, Sudoku \cite{sudoku}, Tetris \cite{tetris}, Go \cite{go}, and Chess \cite{chess}, as well as more recent ones such as Pandemic \cite{pandemic} and Candy Crush \cite{candycrush1, candycrush2}. Because of its popularity, this practice has lead to the emergence of a new field called \emph{algorithmic combinatorial game theory} 
\cite{AlgGameTheory_GONC3,bob+erik09},
in order to distinguish it from other existing fields studying games from different perspectives,
such as `combinatorial game theory' which focuses on the mathematical properties of winning strategies
in combinatorial games \cite{winningways}, or algorithmic game theory which also studies the
algorithmic properties of optimal strategies but in an economic setting \cite{nisan2007algorithmic}.

In this paper we study the computational complexity of a cooperative card game called Hanabi. Designed by Antoine Bauza and published in 2010, the game has received several tabletop game awards (including the prestigious {\em Spiel des Jahres} in 2013~\cite{spiel}). In the game the players simulate a fireworks show\footnote{the word {\em hanabi} means fireworks in Japanese.}, playing cards of different colors in increasing order.
As done previously for other multiplayer card games~\cite{ddhuuu-uhesp-14, LampisM14}, we study a single-player version of Hanabi and show that, even in this simplified model, the game is computationally intractable in general, while it becomes easy under very tight constraints.

\subsection{Rules of the Game}
Hanabi is a multi-player, imperfect-information and cooperative game. 
This game is played with a deck of fifty cards. Each card has two attributes: a \emph{value} (a number from 1 to 5) and a \emph{color} among five possible colors (red, yellow, green, blue and white). Thus, there are 25 different value-color combinations in total, but almost all combinations appear more than once in the deck. 
Players must cooperate in order to create five independent piles of cards from 1 to 5 in increasing order in each of the five colors.

One distinctive feature of the game is that players cannot see their own cards while playing: each player holds her cards so that they can be seen by other players (but not herself). At any given time, a player can hold only a small number of cards in hand (4 or 5 depending on the number of players) drawn at random from the deck. During their turn, players can do one of the following actions: play a card from their hand and draw a new card, discard a card from their hand to draw a new one, or give a hint to another player on what type of cards this other player is holding in hand.  
See~\ref{app:rules} for the exact rules of the Hanabi game, or~\cite{AntoineBauza10,BGG} for more information on the game.

\subsection{Related Work}
Several card games have already been studied from the computational complexity viewpoint \cite{ddhuuu-uhesp-14, LampisM14, durak, trick, magic}. One common element of virtually all such games is that the total description complexity can essentially be bounded by the number of cards (a constant), thus algorithmic questions can, technically speaking, be answered in constant time by an exhaustive search approach. Having said that, in order to study the algorithmic properties of a card game, we need to define an unbounded version of that game, where the complexity is expressed as a function of the number of cards in the deck. 

Another feature that is often present in this type of card-games is some form of randomness. Most commonly, the deck is shuffled so that the exact order in which cards arrive is unknown. This makes the game more fun, but it becomes hard to analyze from a theoretical point of view. Thus, in many cases we simplify our model and assume a {\em perfect information} setting in which everything is known. For example, for a deck of cards, even though every ordering of cards in the deck is a possible input, we assume the ordering is known to the players.
This simplification is quite common when studying card games and also meaningful when one is proving hardness results: even with perfect information, most games turn out to be difficult.

To give some concrete examples, the card game UNO was shown to be NP-hard even for a single player \cite{ddhuuu-uhesp-14}. Even more surprisingly, the popular trading card game {\em Magic: The Gathering} is Turing complete~\cite{magic}. That is, it can simulate a Turing machine (and in particular, it can simulate any other tabletop or card game). All of the above reductions assume perfect information.

There is little previous research studying the algorithmic aspects of Hanabi. Most of the existing research~\cite{cddkrt-mph-15,o-shehoacgii-15} proposes different strategies so that players can share information and collectively play as many cards as possible. Several heuristics are introduced, and compared to either experienced human players or to an optimal play sequences (assuming all information is known). 

Our approach diverges from the aforementioned studies in that it does not focus on information exchange through hints. We show that, even if we forego its trademark feature, the hidden information, the game is still intractable, which means that there is an intrinsic difficulty in Hanabi beyond information exchange. In fact we show hardness for a simplified solitaire version of the game where the single player has complete information about which cards are being held in her hand as well as the exact order in which cards will be drawn from the deck. 

\subsection{Model and Definitions}
In our unbounded model, we represent a card of Hanabi as an ordered pair $(a_i, k_i)$, where $a_i \in \{1,\dots , n\}$ is its \emph{value} and $k_i \in \{1,\dots , c\}$ its \emph{color} (in the original game $n = c = 5$).
The {\em multiplicity} $r$ of cards is the maximum number of times that any card appears in the deck (in the game $r=3$, though some card occur fewer than $3$ times). The whole deck of cards is then represented by a sequence $\sigma$ of $N \le n\cdot c \cdot r$ cards. That is, $\sigma=((a_1,k_1),\dots ,(a_N, k_N))$.  The {\em hand size} $h$ is the maximum number of cards that the player can hold in hand at any point during the game. 

In a game, cards are drawn in the order fixed by $\sigma$. 
During each turn a player normally has three options: play a card from her hand, spend a hint token to give a hint to a fellow player, or discard a card to regain a hint token. After her turn, she draws a new card if she needs to replace a played or discarded one. As our model drops completely the information sharing feature of the game, we can ignore moves that gain or spend hint tokens. Furthermore, since our variation has a unique player who knows in advance the order in which cards appear in the deck, there is no need for the player's hand to be full throughout the game: the player can start with an empty hand and go through the whole deck while storing or playing cards on demand. The three available options when drawing a new card are thus: 
{\em discard} it; {\em store} it for future use; or {\em play} it straightaway. 
If a card is discarded, it is gone and can never be used afterwards. 
If instead we store the card, it is saved and can be played later taking into account the hand limit of~$h$. (Note that since we allow playing the top card of the deck, a hand size of $h$ in our version would be comparable to a hand size of $h+1$ in the original game rules, where each card must be taken in hand before playing.)
Cards can be played only in increasing order for each color independently. 
That is, we can play card $(a_i, k_i)$ if and only if the last card of color $k_i$ 
that was played was $(a_{i-1}, k_i)$, or if $a_i=1$ and no card of color $k_i$ has been played. 
After a card has been played we may also play any cards that we stored in hand 
in the same manner.  
The objective of the game is to play (a single copy of) all cards from $1$ to $n$ in all $c$ colors. 
Whenever this happens we say that 
there is a {\em winning play sequence} for the card sequence $\sigma$. 

Thus, a problem instance of the {\sc Solitaire Hanabi} (simply referred to as {\sc Hanabi} in the rest of this paper) consists of a hand size $h\in \mathbb{N}$ and a card sequence $\sigma$ of $N$ cards (where each card is an ordered pair of a value and a color out of $n$ numbers and $c$ colors, and no card appears more than $r$ times). The aim is to determine whether or not there is a winning play sequence for $\sigma$ that never stores more than $h$ cards in hand. 

\subsection{Results and Organization}
In this paper, we study algorithmic and computational complexity aspects of {\sc Hanabi} 
with respect to parameters  $N$, $n$, $c$, $r$ and $h$. We show that the problem is NP-complete, even if we fix some parameters to be small constants. Specifically, in Section~\ref{sec_hard} we prove that the problem is NP-complete for any fixed values of $h$ and $r$ (as long as $h=1$ and $r\geq 3$ or $h \geq 2$ and $r \geq 2$).

Given the negative results, we focus on the design of algorithms for particular cases. For those cases, our aim is to design algorithms whose running time is linear in $N$ (the total number of cards in the sequence), but we allow worse dependencies on $n$, $c$, and $r$ (the total number of values, colors and multiplicity, respectively). 

In Section~\ref{sec:uni} we give a straightforward $O(N)$ algorithm for the case in which $r=1$ (that is, no card is repeated in $\sigma$). A similar result is later also presented for $c=1$ (and unbounded $r$) in Section~\ref{sec:lazy}. In Section~\ref{sec:generalDP} we give an algorithm for the general problem. Note that this algorithm runs in exponential time (expected for an NP-complete problem). 
The exact running times of all algorithms introduced in this paper  are summarized in Table~\ref{tab_results}.

\begin{table}[htb]
\centering
\begin{tabular}{|c||c|c|c|} \hline
Case Studied & Approach Used & Running Time & Observations \\ \hline 
$r=1$ &Greedy & $O(N)=O(cn)$ &Lemma~\ref{lem:solve_uni} in Sec.~\protect{\ref{sec:uni}} \\
$c=1$ &Lazy & $O(N+n\log h)$ & Theorem~\ref{th_lazy} in Sec.~\ref{sec:lazy} \\ 
General Case &Dynamic Programming & $O(N(h^2+hc)c^hn^{h+c-1})$ & Theorem~\ref{theo_dp} in Sec.~\ref{sec:generalDP} \\ \hline 
$h\ge 2$ and $r \geq 2$ or  &NP-complete &  & Theorem~\ref{theo_nphard} in Sec.~\ref{sec_hard} \\ 
$h = 1$ and $r \geq 3$ & & & \\ \hline
\end{tabular}
\caption{Summary of the different results presented in this paper, 
where $N$, $n$, $c$, $r$ and $h$ are the number of cards, the number of values, 
the number of colors, multiplicity, and the hand size, respectively.}
\label{tab_results}
\end{table}

\section{Unique Appearance}\label{sec:uni}
As a warm-up, we consider the case in which each card appears only once, that is, $r=1$. In this case we have exactly one card for each value and color combination. Thus, $N=cn$ and the input sequence $\sigma$ is a permutation of the values from $1$ to $n$ in the $c$ colors. 
Since each card appears only once, we cannot discard any card in the winning play sequence. In the following, we show that the natural greedy strategy is the best we can do: play a card as soon as it is found (if possible). If not, store it in hand until it can be played. 

From the game rules it follows that we cannot play a card $(a_i, k_i)$ until all the cards from value 1 to $a_i - 1$ of color $k_i$ have been played. Thus, we associate an interval to each card that indicates for how long that card must be held in hand. For any card $(a_i,k_i)$, let  $f_i$ be the largest index of the cards of color~$k_i$ whose value is at most $a_i$ (i.e., $f_i= \max_{j\leq N}\{j \colon k_j=k_i, a_j \leq a_i\}$). Note that we can have that $i = f_i$, in which case all cards of the same color and lower value than the $i$-th card have already been drawn from the deck and in case of a greedy algorithm, already played if the sequence if playable. Hence, the card can be played right from the top of the deck. Otherwise, we have $f_i> i$, and card $(a_i,k_i)$ cannot be played until we have reached card $(a_{f_i},k_{f_i})$.
We associate each index $i$ to the interval $(i,f_i]$ which represents the time interval during which the card must be kept in hand. Let $\mathcal{I}$ be the collection of all nonempty such intervals. Let $w$ be the maximum number of intervals that overlap, that is,  $w=\max_{j\leq N} | \{ (i, f_i] \in \mathcal{I} \colon j\in (i, f_i]\} |$.

\begin{lemma}
\label{lem:solve_uni}
There is a solution to any {\sc Hanabi} problem instance with $r=1$ and hand size $h$ if and only if $w \leq h$. Moreover, a play sequence can be found in $O(N)$ time.
\end{lemma}

\begin{proof}
Intuitively, any interval $(i,j]\in \mathcal{I}$ represents the need of storing card $(a_i,k_i)$ until we have reached card $(a_j,k_j)$. Thus, if two (or more) intervals overlap, then the corresponding cards must be stored simultaneously. By definition of $w$, when processing the input sequence at some point in time we must store more than $h$ cards, which implies that no winning play sequence exists. 

In order to complete the proof we show that the greedy play strategy works whenever $w \leq h$. The key observation is that, for any index $i$ we can play card $(a_i,k_i)$ as soon as we have reached the $f_i$-th card. Indeed, by definition of $f_i$ all cards of the same color whose value is $a_i$ or less have already appeared (and have been either stored or played). Thus, we can simply play the remaining cards (including $(a_i,k_i)$) in increasing order. 

Overall, each card is stored only within its interval. By hypothesis, we have $w\leq h$, thus we never have to store more than our allowed hand size. Furthermore, no card is discarded in the play sequence, which in particular implies that the greedy approach will give a winning play sequence with hand size $h$.

Regarding running time, it suffices to show that each element of $\sigma$ can be treated in constant time. To do this, we need a data structure that allows insertions into $\mathcal{H}$ and membership queries in constant time. The simplest data structure that allows this is a hash table. Since we have at most $h$ elements (out of a universe of size $cn$) it is easy to have buckets whose expected size is constant.

The only drawback of hash tables is that the algorithm is randomized (and the bounds on the running time are expected). If we want a deterministic worst case algorithm, we can instead represent $\mathcal{H}$ with a $c\times n$ bit matrix and an integer denoting the number of elements currently stored. In this data structure, each bit in the matrix is flipped at most twice and inspected at most once. Each bit represents a single card of which there is only one occurrence and is flipped only when that card is added or removed from the hand. The bit is inspected only when the card of the same color and one lower value is played, so we can charge the inspection time to the card that is being played. As a result all the operations associated to the bit matrix will require at most $O(cn) = O(N)$ time in total. 
\end{proof}

\section{Lazy Strategy for One Color}~\label{sec:lazy}
We now study the case in which all cards have the same color (i.e., $c=1$). Note that we make no assumptions on the multiplicity or any other parameters. 
Unlike the last section in which we considered a greedy approach, here we describe a lazy approach that plays cards at the last possible moment. 

We start with an observation that allows us to detect how important a card is. For this purpose we define the concept of a \emph{useless} card. A card $i$ with value $a_i$ is considered useless if there are at least $h+1$ values higher than $a_i$ that do not occur on any cards that appear after $i$ in the deck. Intuitively, if we would play card $i$, then to finish playing all values we must store one card for each of these $h+1$ values in hand as they will not occur after playing (and drawing) $i$. More formally we define a useless card as follows.

\begin{quotation}
\noindent \textbf{Useless card:} For any $i\leq N$, we say that the $i$-th card (whose value is $a_i$) is {\em useless} if there exist $w_1, \ldots, w_{h+1}\in \mathbb{N}$ such that: 
\begin{enumerate}[(i)]
\setlength\itemsep{0em}
\item $a_i < w_1 < \dots <w_{h+1} \leq n$ 
\item $\forall j\in \{i+1, \ldots, N\}$ it holds that $a_j \not\in \{w_1, \ldots, w_{h+1}\}$
\end{enumerate}

\end{quotation}

\noindent
We say $w_1, \ldots, w_{h+1}$ to be the witnesses of the useless card.
Observe that no card of value $n-h$ or higher can be useless (since the $w_i$ values cannot exist) and that the last card is useless if and only if $a_N<n-h$. 

\begin{observation}\label{obs_filterbase}
Useless cards cannot be played in a winning play sequence.
\end{observation}

\begin{proof}
Assume, for the sake of contradiction, that there exists a winning play sequence that plays some useless card whose index is $i$. Since we play cards in increasing order, no card of value equal to or bigger than $a_i$ has been played when the $i$-th card is drawn. By definition of useless cards, the remaining sequence does not have more cards of values $w_1, \ldots, w_{h+1}$. Thus, in order to complete the game to a winning sequence, these $h+1$ cards must all have been stored, but this is not possible with a hand size of $h$. 
\end{proof}

Our algorithm starts with a filtering phase that removes all useless cards from $\sigma$. The main difficulty of this phase is that the removal of useless cards from $\sigma$ may make other cards useless. In order to avoid scanning the input multiple times we use two vectors and a max-heap. The first vector $P$ is computed at the start and does not change throughout the algorithm. 
For each index $i\leq N$, we store the index of the last card before $i$ in $\sigma$ with the same value in a vector $P$ (or $-\infty$ if no such card exists). That is, $P[i]=-\infty$ if and only if $a_j \neq a_i$ for all $j<i$. Otherwise, we have $P[i]=i'$ (for some $i'<i$), $a_i=a_{i'}$, and $a_j \neq a_i$ for all $j\in \{i'+1, \ldots, i-1\}$. As our algorithm progresses it will mark cards as useless, so initially all cards are considered as non-useless. We use a second vector $L$ in which each index $i\leq n$ stores the last card of value $i$ that has not been found useless. This vector will be updated as cards are marked as useless and implicitly removed from the sequence. Since initially no card has been found useless, the value $L[i]$ is initialized to the index of the last card with value $i$ in $\sigma$. Finally, we use a max-heap $\mathcal{HP}$ of $h+1$ elements initialized with the indices stored in $L[n-h], \ldots, L[n]$. 
 
Now, starting with $i=n-(h+1)$ down to $1$ we look for all useless copies of value $i$. The invariant of the algorithm is that for any $j>i$, all useless cards of value $j$ have been removed from $\sigma$ and that vector $L[j]$ stores the index of the last non-useless card of value $j$. The heap $\mathcal{HP}$ contains the lowest $h+1$ indices among $L[j], \ldots, L[n]$ (and since it is a max-heap we can access its largest value in constant time). These values of the cards with indices of $\mathcal{HP}$ will be the smallest possible candidate values for the witnesses $w_1, \ldots, w_{h+1}$ (note that we can extract these values in constant time from $\sigma$). The invariants are satisfied for $i=n-(h+1)$ directly by the way $L$ and $\mathcal{HP}$ are initialized. 

Any card of value $i$ whose index is higher than the top of the heap is useless and can be removed from $\sigma$ (the indices in the heap $\mathcal{HP}$ act as witnesses). Starting from $L[i]$, we remove all useless cards of value $i$ from $\sigma$ until we find a card of value $i$ whose index is smaller than the top of the heap. If no card of value $i$ remains we stop the whole process and return that the problem instance has no solution. Otherwise, we have found the last non-useless card of value $i$. We update the value of $L[i]$ since we have just found the last non-useless card of that value. Finally, we must update the heap $\mathcal{HP}$. As observed above, the value of $L[i]$ must be smaller than the largest value of $\mathcal{HP}$ (otherwise it would be a useless card). Thus, we remove the highest element of the heap, and insert $L[i]$ instead. Once this process is done, we proceed to the next value of $i$. Let $\sigma'$ be the result of filtering $\sigma$ with the above algorithm. 

\begin{lemma}\label{lem_depur}
The filtering phase removes only useless cards and $\sigma'$ contains no useless cards. 
Moreover, this process runs in $O(N+n\log h)$ time. 
\end{lemma}

\begin{proof}
Each time we remove a card from the card sequence, the associated $h+1$ witnesses $w_1, \ldots w_{h+1}$ are present in $\mathcal{HP}$. The fact that no more useless cards remain follows from the fact that we always store the smallest possible witness values.

Now we bound the running time. The heap is initialized with $h+1$ elements, and during the whole filtering phase $O(n)$ elements are pushed. Hence, the heap operations take $O(n\log h)$ time. The vector $P$ does not change during the algorithm and can be computed in $O(N)$ time using a scan from $1$ to $N$ with and $n$-length auxiliary array to store the last index where each value occurs in the sequence. Vector $L$ can be initialized by scanning $\sigma$ once. During the iterative phase we can access the last occurrence of any value by using vector $L$. Once a card is removed, we can update the last occurrence stored in $L$ by simple look-up in $P$. Thus, we spend constant time per card that is removed from $\sigma$ (hence, overall $O(N)$ time). 
\end{proof}

Now we describe the algorithm for our lazy strategy. The play sequence is very simple: we ignore all cards except when a card is the last one of that value present in $\sigma'$. For those cards, we play them if possible or store them otherwise. Whenever we play a card, we play as many cards as possible (out of the ones we had stored). 

Essentially, there are two possible outcomes after the filtering phase. It may happen that all cards of some value were detected as useless. In this case, none of those cards may be played and thus the {\sc Hanabi} problem instance has no solution. Otherwise, we claim that our lazy strategy will yield a winning play sequence.

\begin{theorem}\label{th_lazy}
We can solve a {\sc Hanabi} problem instance for the case in which all cards have the same color (i.e., $c=1$) in $O(N+n\log h)$ time.
\end{theorem}
\begin{proof}
It suffices to show that our lazy strategy will always give a winning play sequence, assuming that the filtered sequence contains at least a card of each value. Our algorithm considers exactly one card of each value from $1$ to $n$. The card will be immediately played (if possible) or stored until we can play it afterwards. Thus, the only problem we might encounter would be the need to store more than $h$ cards at some instant of time. 

However, this cannot happen: assume, for the sake of contradiction, that at some instant of time we need to store a card (whose index is $j$) and we already have stored cards of values $a_{i_1}, \ldots, a_{i_h}$. By construction of the strategy, there cannot be more copies of cards with value $a_{i_1}, \ldots, a_{i_h}$ or $a_j$ in the remaining portion of $\sigma'$. Let $p$ be the number of cards that we have played at that instant of time. Remember that we never store a card that is playable, thus $p+1\not\in\{a_j,a_{i_1}, \ldots, a_{i_h}\}$. In particular, the last card of value $p+1$ must be present in the remaining portion of $\sigma'$. However, that card is useless (the values $\{a_j,a_{i_1}, \ldots, a_{i_h}\}$ act as witnesses), which gives a contradiction. 

Thus, we conclude that the lazy strategy will never need to store more than $h$ cards at any instant of time, and it will yield a winning play sequence. Finally, observe that the play sequence for $\sigma$ is easily found from the winning sequence of $\sigma'$, since vector $L$ stores the last non-useless occurrence of each value.
\end{proof}

\section{General Case Algorithm}
\label{sec:generalDP}

In this section we study the general problem setting, where we consider any number of colors ($c$), values ($n$), occurrences ($r$) and handsize ($h$). 
Recall that this problem is NP-complete, even if the hand size is small (see details in Section~\ref{sec_hard}), hence we cannot expect an algorithm that runs in polynomial time. In the following, we give an algorithm that runs in polynomial time provided that both~$h$ and~$c$ are fixed constants (or exponential otherwise). 
We solve the problem using a dynamic programming approach. To this end we construct a table $DP$ in which each entry represents the maximum number of cards of color $c$ (the last color) that we could have played under several constraints. We group these constraints into three groups as follows:

\begin{itemize}
\setlength{\itemsep}{0em}
\item integer $s \leq N$ represents the number of cards from the sequence $\sigma$ that we have drawn. That is we consider play sequences of the first $s$ elements of $\sigma$.
\item $\mathcal{H}$ is the set of cards that we have stored in hand after the $s$-th card $(a_s,k_s)$ has been processed. We might have no constraints on the cards we have in hand, in which case we simply set $\mathcal{H}=\emptyset$.
\item $p_1,\ldots, p_{c-1}$ with $p_i \leq n$ encode how many cards we have played in each of the first $c-1$ colors, respectively.
\end{itemize}

For the purpose of describing the algorithm we consider $DP$ as a table with parameters $s, \mathcal{H}$ and $p_1, \ldots, p_{c-1}$.
For example, when $c = 3$ and we have $s = 42$, $\mathcal{H} = \{(15,1), (10,2)\}$, $p_1 = 10$ and $p_2 = 4$, then we should interpret $DP[42,\{(15,1), (10,2)\},10,4]=6$ as follows: {\em There is a play sequence over the first $s=42$ elements of $\sigma$ so that we have played exactly $p_1 = 10$ cards of the first color, $p_2=4$ of the second, 6 of the third, and we still have cards $(15,1)$ and $(10,2)$ (those in $\mathcal{H}$) stored in hand. Moreover, there is no play sequence that, under the same constraints, plays $7$ cards in the third color}. 

When $s$ is a small number we can find the solution of an entry by brute force (try all possibilities of discarding, storing or playing the first $s$ cards). This takes constant time since the problem has constant description complexity. Similarly, we have $DP[s,\mathcal{H},p_1, \ldots, p_{c-1}]=-\infty$ whenever $|\mathcal{H}|>h$ (because we need to store more than $h$ cards in hand). In the following we show how to compute the table $DP$ for the remaining cases. For each table entry $DP[s-1,\mathcal{H},p_1, \ldots, p_{c-1}]$ (or $DP[\cdot]$ for short) we find the appropriate value by considering what action to take with the $s$-th card. Remember that we can choose to either \emph{discard}, \emph{play} or \emph{store} the card. For ease of description consider the tables $\mathcal{D}[\cdot], \mathcal{S}[\cdot]$ and $\mathcal{P}[\cdot]$, parameterized the same as $DP[\cdot]$. Then an entry in these tables denotes the maximum number of cards of color $c$ that can be played under the same constraints as in $DP[\cdot]$ and the additional constraints that the $s$-th card is discarded, stored or played respectively. Entries in these tables can be obtained from entries in $DP$ with a lower value for $s$ as follows. (Note that we won't explicitly need to construct these additional tables, as their entries are only needed when computing $DP[\cdot]$.)

$$\mathcal{D}[\cdot]=DP[s-1,\mathcal{H},p_1, \ldots, p_{c-1}] $$
$$ \mathcal{S}[\cdot]= \begin{cases} 
      -\infty& \text{if } (a_s,k_s) \not\in H \\
      DP[s-1,\mathcal{H}\setminus \{(a_s,k_s)\},p_1, \ldots, p_{c-1}]  & \text{otherwise}  
   \end{cases}
$$
$$ \mathcal{P}[\cdot]= \begin{cases} 
      -\infty& \text{if } k_s <c,\, a_s>p_{k_s} \\
      DP[s-1,\mathcal{H}\cup \{(a_s\!+\!1,k_s), \ldots, (p_{k_s},k_s)\}, \\
      \hspace{3cm} p_1, \ldots, p_{k_s-1}, a_s\!-\!1, p_{k_s+1},\ldots, p_{c-1}] & \text{if } k_s <c,\,  a_s \leq p_{k_s} \\
   \max_{t\in \{0, \ldots, h\}} \{a_s\!+\!t \colon DP[s-1,\mathcal{H}\cup \{(a_s\!+\!1,c),\ldots, (a_s\!+\!t,c)\}, \\  \hspace{3cm} p_1, \ldots, p_{c-1}] = a_s\!-\!1\}  & \text{if }k_s=c  
   \end{cases}
$$

Next we prove that the entries in these additional tables indeed help to compute the correct entry for $DP[\cdot]$.

\begin{lemma}
$DP[s,\mathcal{H},p_1, \ldots, p_{c-1}]  (=DP[\cdot]) =\max(\mathcal{P}[\cdot],\mathcal{S}[\cdot],\mathcal{D}[\cdot])$. 
\end{lemma}
\begin{proof}
To prove this we consider the three actions, \emph{discard}, \emph{store} and \emph{play}, and show that the maximum number of cards of color $c$ that we can play corresponds to $\mathcal{D}[\cdot], \mathcal{S}[\cdot]$ and $\mathcal{P}[\cdot]$ respectively.

\begin{description}
\item[$(a_s,k_s)$ is discarded.] When the last card in the play sequence is discarded, the entry of the table is the same as if we only allow the scanning of  $s-1$ cards. Thus, $DP[\cdot]= DP[s-1,\mathcal{H},p_1, \ldots, p_{c-1}]=\mathcal{D}[\cdot]$.

\item[$(a_s,k_s)$ is stored] We claim that in this case $DP[\cdot] = S[\cdot]$. First observe that the $s$-th card must appear in $\mathcal{H}$ as we store it in hand. Then further observe that we have to consider only play sequences in which, after playing the $(s-1)$-th card, a card with value and color $(a_s,k_s)$ is not stored in hand. Indeed, for any winning play sequence that has such a card stored after processing the $s-1$-th card, we can achieve the same result by not storing that card earlier and storing the $s$-th card instead. Therefore it suffices to consider play sequences that end at the $(s-1)$-th card and in which $(a_s,k_s)$ is not stored in hand, as defined in $\mathcal{S}[\cdot]$. 

\item[$(a_s,k_s)$ is played] In this case we claim that $DP[s,\mathcal{H},p_1, \ldots, p_{c-1}]=\mathcal{P}[\cdot]$. We consider the three cases that make up the definition of $\mathcal{P}[\cdot]$. 

\begin{description}
\item[$k_s <c$ and $ a_s>p_{k_s}$] Recall we need to play only up to value $p_{k_s}$ in color $k_s$ and the $s$-th card is of higher value. Therefore, playing this card cannot satisfy the constraint of playing exactly up to value $p_{k_s}$, and the entry should be set to $-\infty$.

\item[$k_s <c$ and $a_s\leq p_{k_s}$] In this case we can safely assume the other colors are not affected. After all, if we can play more cards in any color other than $k_s$, we could have played those earlier. So we need to consider cards of only color $k_s$. Moreover, to play the $s$-th card, we need that cards of this color of value up to $a_s-1$ have already been played. (Again we can safely assume that we need not play these from hand now as they could have been played earlier in that case.) So we are interested in play sequences that consider the first $s-1$ cards and have played up to value $a_s-1$ in color $k_s$. Now to ensure we satisfy the constraint of playing up to value $p_{k_s}$ we need that all cards of value $a_{s}+1$ up to $p_{k_s}$ are already in hand. So to find the correct entry we can simply look in the $DP$-table for the entry where we processed $s-1$ cards, where we have the cards $\mathcal{H} \cup \{(a_s+1, k_s), \ldots, (p_{k_s}, k_s)\}$ in hand and where we played cards up to value $p_i$ in color~$i$ for all $1 \leq i \leq k_s$ and $k_s + 1 \leq i \leq c-1$ and up to value $a_s-1$ in color $p_{k_s}$ as the definition of $\mathcal{P}[\cdot]$ stipulates.

\item[$k_s=c$] This case is intuitively similar to the previous two, but we need to handle it differently because constraints on color~$c$ are encoded differently in the table. As before we can safely assume that we are not playing any cards of other colors together with the $s$-th card. Since we will be playing a card of color $c$ with value $a_c$ we need that a card of color $c$ and value $a_{s}-1$ is already played, but one of value $a_s+1$ is not. In $DP$ that means we are interested in entries in $DP$ of value $a_s-1$. Now the entry for $DP[\cdot]$ should be the highest value in color $c$ that we can play, so we should find the maximum number of cards of color $c$ that we can play after the $s$-th card, which must already reside in hand. To summarize we must find the maximum number of cards we can play in color $c$, where the cards of value $a_s+1$ and higher must reside in hand under the constraint that we only processed the first $s-1$ cards and have played up to value $p_i$ in color $i$ for $1 \leq i \leq c-1$ and up to value $a_s-1$ in color $c$, again precisely as the definition of $\mathcal{P}[\cdot]$ stipulates. 
\end{description}
\end{description}

Thus, for the entry for $DP[\cdot]$ for each of the three actions, discard, store and play, we showed that the maximum number of cards played in color $c$ is as defined in $\mathcal{D}[\cdot], \mathcal{S}[\cdot]$ and $\mathcal{P}[\cdot]$ respectively. Since these are the only valid actions it follows that $DP[\cdot] = \max(\mathcal{P},\mathcal{S},\mathcal{D})$.
\end{proof}

\begin{theorem}\label{theo_dp}
We can solve a {\sc Hanabi} problem instance in $O(N(h^2+hc)c^hn^{h+c-1})$ time using $O(c^hn^{h+c-1})$ space.
\end{theorem}
\begin{proof}
By definition, there is a solution to the {\sc Hanabi} problem instance if and only if its associate table satisfies $DP[N,\emptyset,n,\ldots,n]=n$. Each entry of the table with first parameter $s$ is solved by querying entries with first parameter $s-1$, so we can compute the whole table in increasing order. 

Recall that, for those entries of the table for which the associated set $\mathcal{H}$ has more than $h$ elements the answer is trivially $-\infty$ (since we cannot store that many cards). That means that for the parameter $\mathcal{H}$ we have $\sum_{i\leq h}{nc \choose i}$ sets to consider. In total we will compute $N \cdot \sum_{i\leq h}{nc \choose i} \cdot n^{c-1} = O(Nc^hn^{h+c-1})$ entries. Note that we only need to store entries of at most two columns in the table (first parameter $s$ and $s-1$), so the total storage required is $O(c^hn^{h+c-1})$.

To compute a single entry of the table (say, $DP[s, \mathcal{H}, p_1, \ldots,$ $p_{c-1}]$) we must first compute $\mathcal{D}[\cdot], \mathcal{S}[\cdot]$ and $\mathcal{P}[\cdot]$. We can compute $\mathcal{D}[\cdot]$ and $\mathcal{S}[\cdot]$ with a constant number of queries to the table. For each query we must update the set of cards in hand. As we are adding or removing only one card from $\mathcal{H}$ and its maximum size is $h$ we can do this in $O(\log h)$ time. (We need $O(\log h)$ and not $O(1)$ due to the fact that the cards in $\mathcal{H}$ need to be stored in increasing order, to avoid considering all possible orderings of the same set of cards.) Next we need to find the correct entry in the table. However, the combined space of all table indices is not constant. Specifically, we can represent the hand $\mathcal{H}$ with $O(h)$ words of memory (assuming each card identifier can be stored in $O(1)$ memory) and we need an additional $O(c)$ memory to store a number between $1$ and $n$ for each color parameter $p_1, \ldots, p_{c-1}$. Therefore, the total size of the memory address for any table entry is $O(h+c)$, so we can access it in $O(h+c)$ time.
To compute $\mathcal{P}$ we query up to $O(h)$ entries in the table, resulting in $O(h^2+hc)$ time in total.

Since each entry can be computed in at most $O(h^2+hc)$ time, the total time to compute all entries is $O(N(h^2+hc)c^hn^{h+c-1})$.
\end{proof}

{\bf Remark} In principle, the $DP$ table returns only whether or not the instance is feasible. We note that, we can also find a winning play sequence with standard backtracking techniques.

\section{NP-Hardness (Multiple Colors, Multiple Appearances) }\label{sec_hard}
In this section we prove hardness of the general {\sc Hanabi} problem. As mentioned in the introduction, the problem is NP-complete even if $h$ and $r$ are small constants. Specifically, we will prove the following theorem.

\begin{theorem}\label{theo_nphard}
The {\sc Hanabi} problem is NP-complete for any $r\geq 2$ and $h\geq 2$, as well as for $r\geq 3$ and $h=1$.
\end{theorem}

We first prove the statement for $r=2$, $h=2$ and then show how to generalize it for other values of $r$ and $h$. Our reduction is from {\sc 3-SAT}. Given a {\sc 3-SAT} problem instance with $v$ variables $x_1,\ldots, x_v$ and $m$ clauses $W_1, \ldots, W_m$, we construct a {\sc Hanabi} sequence $\sigma$ with $2v+1$ colors, $n=6m+2$, $r=2$, $h=2$ and $N\le 2(2v+1)(6m+2)$.
Before discussing the proof, we provide a birds-eye view of the reduction. The generated sequence will have a variable gadget $V_i$ for each variable $x_i$ and a clause gadget $C_j$ for each clause $W_j$. The general idea is that each variable will be represented by two colors, one corresponding to a true assignment and one to a false assignment of the variable. We then ensure that at the start of the game in each color, only one of these two colors can progress to a certain value (not both). The clause gadgets then ensure that only if at least one of its literals is true, then all variable colors can make equal progress, otherwise, some color will not be able to progress. At the end, all colors can play their final cards if and only if each color has made progress during each clause gadget. Now we describe our gadgets more precisely.
\medskip

\noindent \textbf{Handsize alteration gadgets.} In some parts of our construction we may want to ensure that the player cannot store in its hand cards of $2v$ variable colors between gadgets. To enforce this we use a dummy color $d$ and we can add cards $(i+1,d), (i+2,d), (i,d)$ to the sequence when the dummy color has progressed up to value $(i-1)$. The dummy color will have only one card for each value, hence each card must be either stored or played. In order to play all cards in the dummy color we must store $(i+1,d)$ and $(i+2,d)$ in hand until we encounter $(i,d)$, effectively preventing from other cards being stored. We call this a {\em hand dump} gadget. By placing cards between $(i+1,d)(i+2,d)$ and $(i,d)$ we can ensure the handsize is reduced to zero for that interval. A handsize of one is achieved in a similar way by using two cards ($(i+1,d)$ and $(i,d)$) instead. We call this the {\em hand reduction} gadget.

\noindent \textbf{Assumptions.} From now on, for ease of description, we only consider play sequences that play all cards in the dummy color. Since each value appears exactly once in the dummy color, if any of them is not played then the play sequence cannot be a winning one. Similarly, we assume cards are played as soon as possible. In particular, if the card that is currently being scanned is playable, then it will be immediately played. We can make this assumption because holding it in hand or discarding it when it could be played is never beneficial. 

\noindent \textbf{Variable gadget.}
 For any $i\leq v$, variable gadget $V_i$ is defined as the sequence $V_i=2,\overline{2},1,3,4,5,\overline{1},\overline{3},\overline{4},\overline{5}$, where overlined values are cards of color $2i$, whereas the other cards have color $2i-1$, the numbers indicate their value.
The first part of the {\sc Hanabi} problem instance $\sigma$ simply consists of the concatenation of all gadgets $V_1$, $\ldots,$ $V_v$, adding card $(2,2v+1)$ in the very beginning and card $(1,2v+1)$ in the very end of the sequence, as to form a hand reduction gadget (see Figure~\ref{fig:variablegad}). We call this sequence $\sigma_1$.

\begin{figure}
  \centering
    \begin{tabular}{p{0.3cm}cccp{0.3cm}c}
    2&\textcolor{red}{2} \textcolor{blue}{2} \textcolor{red}{1} \textcolor{red}{3} \textcolor{red}{4} \textcolor{red}{5} \textcolor{blue}{1} \textcolor{blue}{3} \textcolor{blue}{4} \textcolor{blue}{5} 
    &\textcolor{green}{2} \textcolor{cyan}{2} \textcolor{green}{1} \textcolor{green}{3} \textcolor{green}{4} \textcolor{green}{5} \textcolor{cyan}{1} \textcolor{cyan}{3} \textcolor{cyan}{4} \textcolor{cyan}{5} 
    &\textcolor{magenta}{2} \textcolor{Orange}{2} \textcolor{magenta}{1} \textcolor{magenta}{3} \textcolor{magenta}{4} \textcolor{magenta}{5} \textcolor{Orange}{1} \textcolor{Orange}{3} \textcolor{Orange}{4} \textcolor{Orange}{5} &1\\
    d&1 2 1 1 1 1 2 2 2 2 &3 4 3 3 3 3 4 4 4 4&5 6 5 5 5 5 6 6 6 6&d\\
    \end{tabular}    
    \caption{Sequence $\sigma_1$ for a SAT instance with three variables. The upper row represents the values of the cards whereas the lower one represents the color of each card. Note that the dummy cards to reduce hand size are also added (color ``d'' stands for dummy color $2v+1$).}
    \label{fig:variablegad}
\end{figure}

\begin{lemma}\label{lem_init}
There is no valid play sequence of $\sigma_1$ that can play cards of value $2$ of colors $2i-1$, $2i$ and the dummy color $2v+1$. This statement holds for all $i\leq v$. 
\end{lemma}
\begin{proof}
Assume, for the sake of contradiction, that there exists some $i \leq v$ and a play sequence for which we can play the three cards. In order to play card $(2,2v+1)$ we need to store it in the very beginning of the game enforcing the hand reduction gadget for the duration of the variable assigning phase, thus temporarily reducing the hand-size to one.

Further notice that each card appears exactly once in $\sigma_1$ (that is, the multiplicity of this part is equal to $1$), and that the cards of color $2i-1$ and $2i$ only appear in gadget $V_i$. More importantly, the value $2$ in both colors appears before the value $1$ in the respective colors. In particular, both must be stored before they are played. However, this is impossible, since we have decreased the hand size through the hand reduction gadget.
\end{proof}

Thus, the best we can do after scanning through all variable gadgets is to play five cards of either color $2i-1$ or $2i$ and only one card of the other color. This choice is independent for all $i\leq v$, hence we associate a truth assignment to a play sequence as follows: we say that variable $x_i$ is set to {\em true} if, after $\sigma_1$ has been scanned, the card $(5,2i-1)$ has been played, false if $(5,2i)$ has been played. For well-definement purposes, if neither $(5,2i-1)$ or $(5,2i)$ is played we simply consider the variable as {\em unassigned} (and say that an unassigned variable never satisfies a clause). Again, this definition is just used for completeness since, as we will see later, no variable will be unassigned in a play sequence that plays all cards.
\medskip

\noindent \textbf{Clause gadget.} We describe the gadget $C_j$ for clause $W_j$. We associate three colors to a clause. Specifically, we associate color $2i-1$ with $W_j$ if $x_i$ appears positive in $W_j$. If $x_i$ appears in negated form, we instead associate color $2i$ with $W_j$. Since each clause contains three distinct literals (each literal associated to a distinct variable), it will be associated to three distinct colors. 

Let $o_j=5(j-1)$. Intuitively speaking, $o_j$ indicates how many cards of each color can be played (we call this the {\em offset}). Our invariant is that for all $i\leq v$ and $j\leq m$, before scanning through the clause gadget associated to $W_j$, there is a play sequence that plays up to $o_j+1$ cards of color $2i-1$ and $o_j+5$ of color $2i$ (or the reverse) and no play sequence can exceed those values in any color. Observe that the invariant is satisfied for $j=1$ by Lemma~\ref{lem_init}.

For clause gadget $C_j$ we want to ensure that, if the assignment makes the clause true, then each color that had advanced to at least $o_j+5$ before this gadget can now advance to $o_{j+1} + 5 = o_j+10$ and every color that had advanced to $o_j+1$ can now advance to $o_{j+1}+1 = o_j+6$. Otherwise, one of the associated colors should not be able to advance beyond $o_j+2$.

The clause gadget is constructed as follows. First we add the sequence $o_j+6,o_j+7,o_j+8,o_j+9,o_j+10$ for the three colors associated to $W_j$, ensuring that, among clause literals, those that had advanced to $o_j+5$ can now advance to $o_j+10$.
Then we append the sequence with $o_j+5,o_j +6,o_j+7,o_j+8,o_j+9,o_j+10,o_j+2,o_j+3,o_j+4$ in all colors corresponding to non-literals (except the dummy color). This allows a five-card advancing in those other colors, that is, colors where we had played a value of $o_j+5$ should advance to $o_j+10$ whereas those where we had played up to value $o_j+1$ should advance to $o_j+6$ (observe that the $(o_j+5)$ and $(o_j+6)$ cards can be stored in hand).
Next, we add three cards of the dummy color forming a hand dump gadget.
Finally, we add the sequence $o_j+3,\overline{o_j+3},\overline{\overline{o_j+3}}, o_j+2,o_j+4,o_j+5,o_j+6, \overline{o_j+2},\overline{o_j+4},\overline{o_j+5},\overline{o_j+6}$, $\overline{\overline{o_j+2}}, \overline{\overline{o_j+4}},\overline{\overline{o_j+5}},\overline{\overline{o_j+6}}$ in the three colors associated to $W_j$. As before, no-, single- and double-overline on the numbers indicate the three different colors. See Figure~\ref{fig:clausegad} for an illustration. 

\begin{figure}[tb]
  \centering
\resizebox{\textwidth}{!}{
\begin{tabular}{p{1.5cm} c p{3.6 cm} p{2 cm}ccccc}
& \multicolumn{7}{c}{$C_1$} \\
& \multicolumn{7}{c}{$\overbrace{\rule{14.5cm}{0pt}}$} \\
4  5  3 & 6 7 8 9 10 &  5 6 7 8 9 10 2 3 4 & 7  8  6 & \textcolor{red}{3}  \textcolor{cyan}{3}  \textcolor{magenta}{3} & \textcolor{red}{2}  \textcolor{red}{4}  \textcolor{red}{5}  \textcolor{red}{6} & \textcolor{cyan}{2}  \textcolor{cyan}{4}  \textcolor{cyan}{5}  \textcolor{cyan}{6} & \textcolor{magenta}{2}  \textcolor{magenta}{4}  \textcolor{magenta}{5}  \textcolor{magenta}{6} \\
 d  d  d &1,4, and 5 & 2, 3, and 6 &d d d&1 4 5&1 1 1 1&4 4 4 4&5 5 5 5\\
& \multicolumn{7}{c}{$C_2$} \\
& \multicolumn{7}{c}{$\overbrace{\rule{15cm}{0pt}}$} \\
10  11  9 & 11 12 13 14 15  &  10 11 12 13 14 15 7 8 9 & 13  14  12 & \textcolor{red}{8}  \textcolor{green}{8}  \textcolor{Orange}{8} & \textcolor{red}{7}  \textcolor{red}{9}  \textcolor{red}{10}  \textcolor{red}{11} & \textcolor{green}{7}  \textcolor{green}{9}  \textcolor{green}{10}  \textcolor{green}{11} & \textcolor{Orange}{7}  \textcolor{Orange}{9}  \textcolor{Orange}{10}  \textcolor{Orange}{11}\\ d  d  d & 1,3, and 6  & 2, 4, and 5 &d d d&1 3 6&1 1 1 1&3 3 3 3&6 6 6 6
\end{tabular}
}
    \caption{Sequence $\sigma_2$ for a SAT instance with three variables $x_1,x_2,x_3$ and two clauses $W_1 = (x_1 \lor \lnot x_2 \lor x_3), W_2 = (x_1 \lor x_2 \lor \lnot x_3$). Colors 1, 4, 5 are associated to $W_1$ and colors 1, 3, 6 are associated to $W_2$. The upper row represents the values of the cards whereas the lower one the color of each card. Note that the dummy cards to obtain independence between/inside gadgets are also added (color ``d'' stands for dummy color $2v+1$).}
\label{fig:clausegad}
\end{figure}

Let $\sigma_2$ be the result of concatenating all clause gadgets in order, where before each $C_j$ we add three cards of the dummy color forming a hand dump gadget to make sure that no card from one gadget can be saved to the next one (see Figure~\ref{fig:overall}). Further let $\sigma'$ be the concatenation of $\sigma_1$ and $\sigma_2$.  
\begin{figure}
  \centering
    \begin{tabular}{p{0.5cm}ccp{0.5cm}p{0.5cm}ccp{0.5cm}ccccp{1cm}cc}
    $\sigma_1$&4&5&3&$C_1$&10&11&9 & $C_2$ & \ldots & $6m - 2$ & $6m-1$ &$6m-3$ & $C_m$   \\
    \end{tabular}
    \caption{Overall picture of the reduction. All cards depicted have dummy color (and are only used to obtain independence between gadgets).}
\label{fig:overall}
\end{figure}
We must show that, when scanning a clause gadget that is satisfied, we can make the correct amount of progress in each color (which is five additional values played). We show this first for colors corresponding to literals not belonging in the clause in Lemma~\ref{lem_assigned_color_1}~and~\ref{lem_assigned_color_2}. Then we proceed with the remaining three colors in Lemma~\ref{lem_invariant}. For each of the three cases note that any card played while processing a clause gadget $C_j$ must occur within that clause due to the presence of hand dump gadgets between any two clause gadgets.

\begin{lemma}\label{lem_assigned_color_1}
Let $k \leq 2v$ be any color for which we have played exactly up to value $o_j+5$ before processing $C_j$ for any $j \leq m$. Then we can play exactly up to $o_j+10 = o_{j+1}+5$ in color $k$ during the processing of $C_j$.
\end{lemma}
\begin{proof}
Observe that $C_j$ contains the sequence $o_j+6,o_j+7,o_j+8,o_j+9,o_j+10$ in consecutive fashion in all colors. Thus, the five cards can be played without having to store anything in hand. Also note that a sixth card cannot be played since $o_j+11$ is not present in any color in $C_j$.
\end{proof}

\begin{lemma}\label{lem_assigned_color_2}
Let $k \leq 2v$ be any color that is not associated with $C_j$ and for which we have played exactly up to value $o_j+1$ before processing $C_j$ for any $j \leq m$. Then we can play exactly up to $o_j+6 = o_{j+1}+1$ in color $k$ during the processing of $C_j$.
\end{lemma}
\begin{proof}
In this case, the cards of color $k$ appear in the following order: $o_j+5,o_j +6,o_j+7,o_j+8,o_j+9,o_j+10,o_j+2,o_j+3,o_j+4$. 
It is straightforward to verify that if we are allowed to store only two cards the best we can do is to store values $o_j+5$ and $o_j+6$ until $o_j+2,o_j+3,o_j+4$ have been played. Similarly $o_j+7$ cannot be saved as this would require storing three cards in hand.
\end{proof}

The remaining case is that a color $k$ is associated to $W_j$ and only $o_j+1$ cards have been played. Recall that, by the way in which we associated variable assignments and play sequences, this corresponds to the case that the assignment of variable $x_{\lceil k/2 \rceil}$ does not satisfy the clause $W_j$. We now show that five cards of color $k$ will be playable if and only if at least one of the other two variables satisfies the clause. 

\begin{lemma}\label{lem_invariant}
Let $C_j$ be the clause gadget associated to $W_j$ (for some $j\leq m$). We can play five cards in each of the three colors associated to $W_j$ if and only if we have played a card of value $o_j+2$ in at least one of the three associated colors before $W_j$ is processed.  Moreover, we can never play more than five cards in the three colors associated to $W_j$.
\end{lemma}
\begin{proof}
Recall that we consider a literal to be true if its color advanced to $o_j+5$ before processing $C_j$ and false if it advanced to $o_j+1$ (if it has not advanced even up to $o_j+1$ then the formula is unsatisfiable as described later). If any literal was true, then it can advance to $o_j+10$ due to the first sequence of cards from $o_j+6$ to $o_j+10$. Otherwise, it will not have advanced at all yet and due to the hand dump gadget we cannot store cards from the first part of $C_j$. Now we can show the claimed equivalence in both directions.

If at least one literal is true it does not need any cards from the second segment to advance and will have already advanced to $o_j+10$ by the time we reach values $o_j+3$ in the gadget. For the other two literals we have that either they will have also already advanced (if they were also true), or they can advance by storing the value $o_j+3$ in hand for each non-advanced color and playing all values from $o_j+2$ to $o_j+6$. 

However, if all three literals are false, then none will have advanced past $o_j+1$ when entering the second part of $C_j$. Since we can only store $o_j+3$ in hand for two colors, one color won't be able to advance beyond $o_j+2$.
\end{proof}

From the above results we know that after we scanned through $\sigma '$ we can play at least up to value $o_m+6$ (in half of the colors we can play up to value $o_m+10$) if and only if the variable assignment created during the variable assignment phase satisfied all clauses. For the dummy color, we used a hand reduction gadget and two hand dump gadgets per clause, thus $6m+2$ cards  have been played. We append to $\sigma '$ cards with values $o_m+6$ to $6m+2$ in increasing order in all colors (except the dummy color). Let $\sigma$ be the resulting sequence.

\begin{theorem}\label{theo_equival}
There is a valid solution of {\sc Hanabi} for $\sigma$ $(r=2)$ and $h=2$ if and only if the associated problem instance of {\sc 3-SAT} is satisfiable. 
\end{theorem}

\begin{proof}
If the associated problem instance of 3-SAT is satisfiable, then there exists a truth assignment satisfying all clauses, and then by Lemmas~\ref{lem_init},~\ref{lem_assigned_color_1},~\ref{lem_assigned_color_2} and~\ref{lem_invariant}, we can play all colors up to the card $6m+2$ from~$\sigma$.
If the associated problem instance of 3-SAT is not satisfiable then, for any truth assignment, there should be one or more clauses that are not satisfied. Let $j$ be the index of the first clause that is not satisfied by some truth assignment. By Lemma~\ref{lem_invariant}, we will not be able to play a card of value $o_j+3$ in one of the three colors associated to $C_j$. Since the smallest value in any of the next gadgets is $o_j+7$, no more cards can be played in that color. In particular, there cannot be a solution for this {\sc Hanabi} problem instance.
\end{proof}

\subsection{Modifications to the reduction}
The above reduction can be easily constructed in polynomial time. In order to complete the proof of Theorem~\ref{theo_nphard} we must show how to adapt the hardness for other values of $r$ and $h$. In this section we introduce these and other modifications to our reduction.

\paragraph{Fixed multiplicity and larger hand size} Adapting the construction for larger values of $r$ is easy: if we want to have exactly $r>2$ copies of each card, it suffices to place $r-2$ additional copies of each card at some position in which it cannot affect our reduction (for example, the first time a card $c$ appears we place $r-2$ identical copies). It is never useful to play (or store) more than one copy of the same card, so overall the reduction is unaffected. Similarly, if $h>2$ we use a hand reduction gadget to make sure that the hand size is  exactly $2$ for the interval in which $\sigma$ is processed. 

\paragraph{Small hand size} In order to make our construction work for $h=1$, we slightly modify the clause gadget. Recall that, for the colors associated to a clause gadget $W_j$ we place the cards in the following order: $o_j+3,\overline{o_j+3},\overline{\overline{o_j+3}}, o_j+2,o_j+4,o_j+5,o_j+6, \overline{o_j+2},\overline{o_j+4},\overline{o_j+5},\overline{o_j+6}$, $\overline{\overline{o_j+2}}, \overline{\overline{o_j+4}},\overline{\overline{o_j+5}},\overline{\overline{o_j+6}}$. Immediately after these cards we now insert twelve additional cards as follows: $o_j+4,\overline{o_j+4},\overline{\overline{o_j+4}}, 
o_j+3,o_j+5,o_j+6, 
\overline{o_j+3},\overline{o_j+5},\overline{o_j+6}$, $\overline{\overline{o_j+3}},\overline{\overline{o_j+5}},\overline{\overline{o_j+6}}$. 

Intuitively speaking, this gadget almost duplicates the original one. Thus, making one pass at the original gadget (with hand size two) is the same as making a pass at this ``doubled'' gadget with hand size one. Similar to the previous analysis, it can be seen that when we scan a modified clause gadget, we can play 5 cards of all colors if and only if that clause is satisfied (otherwise, for one color we won't be able to play more than two cards). Note that this gadget increases the number of repetitions of each card to three, so the case $h=1$ and $r=2$ remains open.

\paragraph{Constant number of values} The construction can also work if $n$ is fixed to be a small constant, say $30$. In this case, we reduce from a variation of 3-SAT called 3-SAT$\kappa$: in addition to the usual 3-SAT constraints, an instance of this problem has the additional constraint that each variable appears in exactly $\kappa$ clauses (and a variable appears in a clause no more than once). Finding an assignment that satisfies such an instance is known to be NP-complete even for $\kappa=5$~\cite{Feige}. 

When examining a clause gadget, instead of advancing in all colors, we just advance 5 numbers in the three colors associated to the literals belonging in that clause (and the three belonging to the complements): specifically, at the $i$-th appearance of variable $x_j$ we can progress up to $5i+1$ and $5i+5$ in the colors $2j$ and $2j+1$ (as usual, which of the two progresses more depends on the variable assignment), but we can only progress in the three pairs of colors if at least one of the variables satisfies the clause. Note that the usual construction also uses $\Omega(m)$ values for the dummy color. We can reduce it to a constant number of values by using a different dummy color each time instead. This will create $O(m)$ dummy colors, each one having a constant number of values as desired.

\paragraph{Larger number of players} The hardness also extends to the case in which we have $p$ players. With more than one player the {\em give a hint} action allows a player to pass (i.e., neither draw nor discard). Thus, if we have enough hints (say, at least $pN$), the game with $p$ players each with a hand size of $h$ can be seen as  
 a single player instance, in which the single player can hold $ph$ cards in hand.

\section{Conclusions}
In this paper, we studied the complexity of a single player, perfect-information version of Hanabi and we showed that it is NP-hard (completeness easily follows), even if the hand size and number of card repetitions are small constants. Proving hardness for such a simplistic model uncovers the importance of hand management, something that had been overlooked in previous studies of the game.

Several questions regarding the complexity of the game remain unanswered. For example, is the game still NP-hard if we bound the number of colors instead of the hand size?
Furthermore, as the problem is very rich in parameters, it would be fruitful to study it from a parameterized complexity point of view. 

\section*{Acknowledgements}
The authors would like to thank Erik Demaine for some insightful comments on the problem.


\bibliography{refs-condenced}

\newpage
\appendix
\section{The Rules of Hanabi}\label{app:rules}

In this appendix we introduce the official rules of Hanabi \cite{AntoineBauza10}. 

\subsection*{Game Material}

50 fireworks cards in five colors (red, yellow, green, blue, white): 
\begin{itemize}
\item
10 cards per color with the values 1, 1, 1, 2, 2, 3, 3, 4, 4, 5, 
\item
10 colorful fireworks cards with same values as above 
\item
8 Clock (Note) tokens (+ 1 spare), 
\item
3 Storm (Fuse) tokens. 
\end{itemize}

\subsection*{Aim of the Game}
Hanabi is a cooperative game, meaning all players play together as a team. 
The players have to play the fireworks cards sorted by colors and numbers. 
However, they cannot see their own hand cards, 
and so everyone needs the advice of their fellow players. 
The more cards the players play correctly, the more points they receive when the game ends.

\subsection*{The Game}
The oldest player is appointed first player and sets the tokens in the play area.
The eight Clock tokens are placed white-side-up.
The three Storm tokens are placed lightning-side-down.

Now the fireworks cards are shuffled. Depending on the number of players involved, 
each player receives the following hand:
\begin{itemize}
\item
With 2 or 3 players: 5 cards in hand, 
\item
With 4 or 5 players: 4 cards in hand. 
\end{itemize}

\noindent
Important: For the basic game, the colorful fireworks cards and the spare Clock token(s) 
are not needed. 
They only come in to use for the advanced game.

\noindent
Important: Unlike other card games, players may not see their own hand! 
The players take their hand cards so that the back is facing the player. 
The fronts can only be seen by the other players. 
The remaining cards are placed face down in the draw pile in the middle of the table. 
The first player starts.

\subsection*{Game Play}

Play proceeds clockwise. On a player's turn, they must perform exactly one of the following:

A. Give a hint or

B. Discard a card or

C. Play a card.

\noindent
The player has to choose an action. A player may not pass!

\noindent
Important: Players are not allowed to give hints or suggestions on other players' turns!

\subsubsection*{\rm A. Give a hint}

To give a hint one Clock token must be flipped from its white side to its black side. 
If there are no Clock tokens white-side-up then a player may not choose 
the Give a hint action.
Now the player gives a teammate a hint. They have one of two options:

\begin{description}
\item[1. Color Hint.]
The player chooses a color and indicates to their teammate 
which of their hand cards match the chosen color by pointing at the cards.
Important: The player must indicate all cards of that color in their teammate's hand!
Example: ``You have two yellow cards, here and here.''
Indicating that a player has no cards of a particular color is allowed!
Example: ``You have no blue cards.''

\item[2. Value Hint.] 
The player chooses a number value and gives a teammate a hint 
in the exact same fashion as a Color Hint.
Example: ``You have a 5, here.''
Example: ``You have no Twos.''
\end{description}

\subsubsection*{\rm B. Discard a card}
To discard a card one Clock token must be flipped from its black side to its white side. 
If there are no Clock tokens black-side-up then a player 
may not choose the Discard a card action.
Now the player discards one card from their hand (without looking at the fronts 
of their hand cards) and discards it face-up in the discard pile near the draw deck. 
The player then draws another card into 
their hand in the same fashion as their original hand cards, never looking at the front.

\subsubsection*{\rm C. Play a card}
By playing out cards the fireworks are created in the middle of the table. 
The player takes one card from their hand and places it face up in the middle of the table. 
Two things can happen: 
\begin{description}
\item[1. The card can be played correctly.] 
The player places the card face up so that it extends a current firework 
or starts a new firework.
\item[2. The card cannot be played correctly.]
The gods are angry with this error and send a flash from the sky. 
The player turns a Storm tile lightning-side-up. 
The incorrect card is discarded to the discard pile near the draw deck.
\end{description}
In either case, the player then draws another card into their hand in the same fashion 
as their original hand cards, never looking at the front.

\subsection*{The Fireworks}
The fireworks will be in the middle of the table and are designed in five different colors. 
For each color an ascending series with numerical values from 1 to 5 is formed. 
A firework must start with the number 1 and each card played to a firework 
must increment the previously played card by one. 
A firework may not contain more than one card of each value.

\subsection*{Bonus}
When a player completes a firework by correctly playing a 5 card 
then the players receive a bonus. 
One Clock token is turned from black side to white side up. 
If all tokens are already white-side-up then no bonus is received. 
Play then passes to the next player (clockwise).

\subsection*{Ending the Game}

The game can end in three ways:

\begin{enumerate}
\item
The third Storm token is turned lightning-side-up. 
The gods deliver their wrath in the form of a storm that puts an end to the fireworks. 
The game ends immediately, and the players earn zero points.

\item
The players complete all five fireworks correctly. 
The game ends immediately, and the players 
celebrate their spectacular victory with the maximum score of 25 points.

\item
If a player draws the last card from the draw deck, the game is almost over. 
Each player---Including the player who drew the last card---gets one last turn.
\end{enumerate}

Finally, the fireworks will be counted. 
For this, each firework earns the players a score equal to the highest value card 
in its color series.
The quality of the fireworks display according to the rating scale 
of the ``International Association of Pyrotechnics'' is:
\begin{itemize}
\item
0--5: Oh dear! The crowd booed.
\item
6--10: Poor! Hardly applaused.
\item
11--15: OK! The viewers have seen better.
\item
16--20: Good! The audience is pleased.
\item
21--24: Very good! The audience is enthusiastic!
\item
25: Legendary! The audience will never forget this show!
\end{itemize}

\subsection*{Important Notes and Tips}

\begin{itemize}
\item
Players may rearrange their hand cards and change their orientation 
to help themselves remember the information they received. 
Players may not ever look at the front of their own cards until they play them.
\item
The discard pile may always be searched for information.
\item
Hanabi is based on communication---and non-communication---between the players. 
If one interprets the rules strictly then players may not, 
except for the announcements of the current player, talk to each other. 
Ultimately, each group should decide by its own measure
what communication is permitted. Play so that you have fun!
\end{itemize}

\end{document}